\documentclass[12pt]{article}
\usepackage{graphicx,color}
\usepackage{booktabs,multirow}
\usepackage{amsmath,XoohmE}

\def\8{{\times}}

\def\hm{{\hat\m}}

\definecolor{Hey}{rgb}{.9,.05,.4}
\definecolor{plum}{rgb}{.4,0,.6}
\definecolor{Tan}{rgb}{.87,.62,0.42}
\definecolor{Green}{rgb}{0.25,.65,.25}
\definecolor{Blue}{rgb}{0.25,.25,.65}
\definecolor{Plum}{rgb}{.47,0.15,.7}
\definecolor{Red}{rgb}{.75,0.15,0.25}
\definecolor{Cyan}{rgb}{0.4,0.6,.95}
\definecolor{yel}{rgb}{.9,.8,0}
\definecolor{gray}{rgb}{.5,.5,.5}
\definecolor{tan}{rgb}{.87,.62,0.42}

\def\sS{\mathscr{S}}
\def\4#1#2{G_#1^{\sss(#2)}}
\def\I#1{\intertext{#1\vspace{-4mm}}}
\def\stk#1{{\scriptsize\shortstack{#1}}}
\def\Stk#1{\vC{\stk{#1}}}
\def\1{\vphantom{I^1_1}}

\def\ba{\mathbf{a}}

\def\vC#1{\vcenter{\hbox{\hss#1\hss}}}

\def\Ft#1{\,\footnote{#1}}
\newdimen\parshift\parshift=\parindent
\catcode`@=11
 \long\def\@footnotetext#1{\insert\footins{\reset@font\footnotesize\interlinepenalty%
  \interfootnotelinepenalty\splittopskip\footnotesep\splitmaxdepth\dp\strutbox%
   \floatingpenalty\@MM\hsize\columnwidth\addtolength{\hsize}{-2\parindent}
    \@parboxrestore\protected@edef\@currentlabel{\csname p@footnote\endcsname\@thefnmark}
      \color@begingroup
       \@makefntext{\rule\z@\footnotesep\ignorespaces#1\@finalstrut\strutbox}
        \color@endgroup}}
 \long\def\@makefntext#1{\hglue\parshift
                         \vbox{\noindent\hb@xt@0em{\hss\@makefnmark}#1}}
\catcode`@=12
 \font\rOpe=cmsy10                        
 \def\ktl{{\hbox{\rOpe\char'170}}}        
 \def\kbl{{\hbox{\rOpe\char'170}}}        
 \def\kcr{{\reflectbox{\rOpe\char'170}}}        
 \def\ktr{{\reflectbox{\rOpe\char'170}}}        
 \def\kbr{{\reflectbox{\rOpe\char'170}}}        
 \def\Border{\vbox{\hsize0pt
        \setlength{\unitlength}{1mm}
        \newcount\xco
        \newcount\yco
        \xco=-21
        \yco=12
        \begin{picture}(0,0)(-7.5,0)
        \put(\xco,\yco){$\ktl$}
        \advance\yco by-1
        {\loop
        \put(\xco,\yco){$\kcr$}
        \advance\yco by-2
        \ifnum\yco>-240
        \repeat
        \put(\xco,\yco){$\kbl$}}
        \xco=170
        \yco=12
        \put(\xco,\yco){$\ktr$}
        \advance\yco by-1
        {\loop
        \put(\xco,\yco){$\kcr$}
        \advance\yco by-2
        \ifnum\yco>-240
        \repeat
        \put(\xco,\yco){$\kbr$}}
        \put(-19.5,13){\scalebox{.598}{State University of New York
            Physics Department|University of Maryland Center for
            String and Particle  Theory \&\ Physics Department|%
            Delaware State University DAMTP}}
        \put(-19.5,-241.5){\scalebox{.728}{University of Washington
            Mathematics Department|Pepperdine University Natural
            Sciences Division|Bard College Mathematics
            Department}}
        \end{picture}
        \par\vskip-8mm}}
\definecolor{UMred}{rgb}{.9,.05,.2}
 \def\UMbanner{\vbox{\hsize0pt
        \setlength{\unitlength}{.4mm}
        \thicklines
        \begin{picture}(0,0)(-30,-10)
        \put(165,16){\line(1,0){4}}
        \put(170,16){\line(1,0){4}}
        \put(180,16){\line(1,0){4}}
        \put(175,0){\line(1,0){4}}
        \put(180,0){\line(1,0){4}}
        \put(185,0){\line(1,0){4}}
        \put(169,0){\line(0,1){16}}
        \put(170,0){\line(0,1){16}}
        \put(179,0){\line(0,1){16}}
        \put(180,0){\line(0,1){16}}
        \put(184,0){\line(0,1){16}}
        \put(185,0){\line(0,1){16}}
        \put(169,16){\oval(8,32)[bl]}
        \put(170,16){\oval(8,32)[br]}
        \put(179,0){\oval(8,32)[tl]}
        \put(185,0){\oval(8,32)[tr]}
        \end{picture}
        \par\vskip-6.5mm
        \thicklines}}

\SfTitles 
 \allowdisplaybreaks

\begin{document}
\thispagestyle{empty}
\vbox{\Border\UMbanner}
 \noindent
 \today\hfill{
              UMDEPP 07-012/SUNY-O/663 
 }
 \begin{center}
{\LARGE\sf\bfseries\boldmath
 Relating Doubly-Even Error-Correcting Codes,\\[-1mm]
  Graphs, and Irreducible Representations of\\[1mm]
   $N$-Extended Supersymmetry\Ft{To appear
   in {\em\/Discrete and Computational Mathematics\/}, Eds: F.~Liu, {\em\/et.~al.\/}, (Nova Science Publishers, Inc.,2008).}}\\*[5mm]
{\sf\bfseries C.F.\,Doran$^a$, M.G.\,Faux$^b$, S.J.\,Gates, Jr.$^c$,
     T.\,H\"{u}bsch$^d$, K.M.\,Iga$^e$ and G.D.\,Landweber$^f$}\\*[2mm]
{\small\it
  \leavevmode\hbox to\hsize{\hss
  $^a$Department of Mathematics,
      University of Washington, Seattle, WA 98105%
  , {\tt  doran@math.washington.edu}\hss}
  \\
  $^b$Department of Physics,
      State University of New York, Oneonta, NY 13825%
  , {\tt  fauxmg@oneonta.edu}
  \\
  $^c$Center for String and Particle Theory,\\[-1mm]
      Department of Physics, University of Maryland, College Park, MD 20472%
  , {\tt  gatess@wam.umd.edu}
  \\
  $^d$Department of Physics \&\ Astronomy,
      Howard University, Washington, DC 20059%
  , {\tt  thubsch@howard.edu}\\[-1mm]
  \leavevmode\hbox to\hsize{\hss
  Department of Applied Mathematics and Theoretical Physics,
      Delaware State University, Dover, DE 19901\hss}
  \\
  $^e$Natural Science Division,
      Pepperdine University, Malibu, CA 90263%
  , {\tt  Kevin.Iga@pepperdine.edu}
  \\
 $^f$Department of Mathematics, Bard College,
     Annandale-on-Hudson, NY 12504-5000%
  , {\tt  gregland@bard.edu}
 }\\[3mm]
{\sf\bfseries ABSTRACT}\\[3mm]
\parbox{145mm}{
Previous work\cite{r6-1} has shown that the classification of indecomposable off-shell representations of $N$-supersymmetry, depicted as {\em Adinkras\/}, may be factored into specifying the {\em topologies\/} available to Adinkras, and then the height-assignments for each topological type. The latter problem being solved by a recursive mechanism that generates all height-assignments within a topology\cite{r6-1}, it remains to classify the former. Herein we show that this problem is equivalent to classifying certain ({\bf1})~graphs and ({\bf2})~error-correcting codes.}
\end{center}
\noindent 
{\em\/Mathematics Subject Classification\/}. Primary: 81Q60; Secondary: 15A66, 16W70\\
{\em\/Keywords\/}: Supersymmetry, Error-Correcting Codes, Graphs

\section{The Statement of the Problem}
 \label{IRS}
 Supersymmetry algebras are a special case of super-algebras, where the odd generators, $Q$, form a spin-$\inv2$ representation of the Lorentz algebra contained in the even part and $\{Q,Q\}$ necessarily contains the translation generators in the even part. Systems exhibiting such symmetry have been studied over more than three decades and find many applications, although experimental evidence that Nature also employs supersymmetry is ironically lacking within high-energy particle physics, where it was originally invented. Nevertheless, as the only known systematic mechanism for stabilizing the vacuum\footnote{Vacuum here denotes the space of absolute ground states in these {\em quantum\/} theories.}, supersymmetry is also a keystone in most contemporary attempts at unifying all fundamental physics, such as string theory and its $M$- and $F$-theory extensions. There, one typically needs a large number, $N\leq32$, of supersymmetry generators, in which case {\em off-shell\/} descriptions, indispensable for a full understanding of the quantum theory, are sorely absent.

\subsection{Reduction to 1 Dimensions}
 This has motivated some of the recent interest%
\cite{rGR0,rGR1,rGR2,rKRT,rT06,rGLPR,rGLP} in the 1-dimensional dimensional reduction of supersymmetric field theories\footnote{We refer to $N$-extended supersymmetry in 1-dimensional time as ``$(1|N)$-supersymmetry''. Besides dimensional reduction of field theories in higher-dimensional spacetime (including here world-sheet theories for superstrings) to their 1-dimensional shadows, $(1|N)$-supersymmetry is also present in the study of supersymmetric wave functionals in {\em any supersymmetric quantum field theory\/}, and so applies to all of them also in this other, more fundamental way.}, whereupon the supersymmetry algebra is specified:
\begin{alignat}{3}
 \big\{\,Q_I\,,\,Q_J\,\big\}&=2\,\d_{IJ}\,H~, &\qquad
 \big[\,H\,,\,Q_I\,]&=0~,\qquad I,J=1,\cdots,N~,\label{e(1|N)}
\end{alignat}
where $H$ is the Hamiltonian, generating the 1-dimensional Poincar\'e group: time-translations.
 This defers the incorporation of the higher-dimensional non-abelian Lorentz algebras till after the classification of $(1|N)$-supermultiplets (=\,representations of $(1|N)$-supersymmetry)\cite{rGR0}. Thereby, the non-abelian nature of higher-dimensional Lorentz algebras gives rise to a collection of obstructions of dimensionally {\em oxidizing\/} some of the $(1|N)$-supermultiplets into their higher-dimensional analogues; the unobstructed ones are then the representations of supersymmetry algebras in 2- and higher-dimensional spacetimes.

\subsection{Adinkras and Representations of $(1|N)$-Supersymmetry}
 In particular, Refs.\cite{rA,r6-1,r6-2} introduce, hone and apply a graphical device, akin to wiring schematics, which can fully encode all requisite details about $(1|N)$-supersymmetry, its action within off-shell supermultiplets and the possible couplings of all off-shell supermultiplets---for all $N$. These graphs, called {\em Adinkras\/}, are closely related to the rigorous underpinning of off-shell representation theory in supersymmetry\cite{r6--1}, but are also intuitively easy to understand and manipulate. We restrict to $(1|N)$-supersymmetry without central charges, and to supermultiplets upon which the supersymmetry acts by a $\ZZ_2^{\times}$-monomial\footnote{Monomial matrices have a single non-zero entry in every row and column; by $\ZZ_2^{\times}$-monomial we mean that the nonzero entries are $\pm1$, forming a multiplicative $\ZZ_2$ group.} linear transformation.
 
 Adinkras represent component bosons in a supermultiplet as white nodes, fermions as black. A white and a black node are connected by an edge, drawn in the $I^{\rm th}$ color if the $I^{\rm th}$ supersymmetry transforms the corresponding component fields one into another. A sign/parity degree of freedom\cite{rA} in the supersymmetry transformation of a component field into another is represented by solid {\em vs\/}.\ dashed edges.
 In the natural units ($\hbar=1=c$), all physical fields have a definite engineering dimension defined up to an overall additive constant, and we accordingly stack the nodes at heights that reflect the engineering dimensions of the corresponding component fields:
 \begin{equation}
 \vC{\includegraphics[height=30mm]{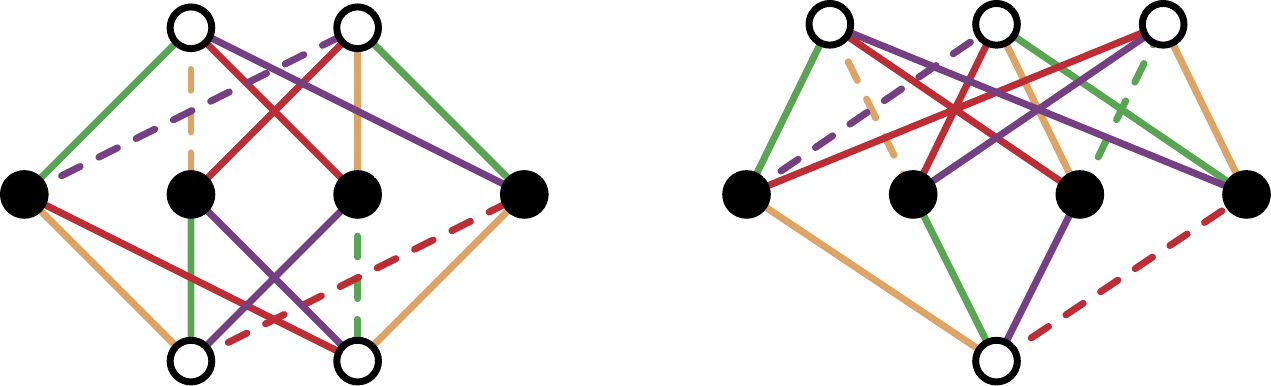}}
 \label{eB242A}
\end{equation}
These two Adinkras depict two $(1|4)$-supermultiplets, consisting of four bosons and four fermions; the action of each of the four supersymmetry generators is represented by the edges colored in one of the four colors. The two top white nodes in the left-most Adinkra correspond to bosons the engineering dimension of which is 1 mass-unit more than those of the bosons corresponding to the bottom two; the engineering dimension of the fermions, depicted by the black nodes in the mid-level, is half-way between these. By contrast, in the right-most Adinkra, there is only one boson with the lowest engineering dimension, and three have 1 mass-unit more than that. The supersymmetry transformation rules may be read off the Adinkras\eq{eB242A} straightforward manner, but we defer the details to the literature\cite{rA,r6-1}; suffice it here to draw the Reader's attention to the fact that every quadrangle must contain an odd number of dashed edges. This corresponds to the fact that the supersymmetry generators, $Q_I$, {\em anticommute\/}.
\ping

In general, then, $(1|N)$-supermultiplets are then obtained, in terms of Adinkras, using standard techniques in linear algebra, by: ({\bf1})~specifying $(1|N)$-sypersymmetry preserving maps between direct sums of Adinkras, including the ``0'' (empty) Adinkra, and ({\bf2})~specifying the kernels and cokernels of such maps, iterating such constructions as needed; we return to this in section~\ref{s:More}.

\subsection{Adinkra Topologies}
 \label{s:AT}
Given any Adinkra, let {\em topology\/} denote the overall connectivity of the nodes within an Adinkra, including the choice of dashed {\em vs\/}.\ full edges but up to sign-redefinition of every node\footnote{By changing the sign of a node, \ie, the component field corresponding to it, every edge connected to it will flip from dashed to full and {\em vice versa\/}.}. Let the collection of all Adinkras that share the same topology be called a {\em family\/}.
 Ref.\cite{r6-1} specifies the (superdifferential) operators that implement a recursive procedure whereby the nodes may be repositioned in all possible height arrangements, starting from any one of them, while maintaining the topology. This procedure thus reconstructs the entire family of Adinkras (and supermultiplets which they represent) from any one of them.

It then remains to specify all the possible Adinkra topologies, given $N$, the number of supersymmetries, and this is the subject of Ref.\cite{r6-3}. We prove therein that all Adinkras have the topology of the form of a $k$-fold iterated $\ZZ_2$-quotient of the $N$-cube, $I^N$, in which the number of supersymmetries is preserved and the supersymmetry action upon the Adinkra remains linear and 1--1 in each iteration. We defer to Ref.\cite{r6-3} for the rigorous proofs, but present here the gist of this result.

\subsubsection{The Nature of the Quotients}
To begin with, we associate to every monomial $\m_j:=Q_1^{a_{j1}}\,Q_2^{a_{j2}}\cdots Q_N^{a_{jN}}$, with $\a_{jl}\in\{0,1\}$, the point $(a_{j1},a_{j2},\cdots,a_{jN})\in[0,1]^N\subset\IR^N$; clearly, there are $2^N$ such monomials and points (nodes). The $I$-labeled edges then connect every point to the one that differs from it only in the $I^{\rm th}$ coordinate, \ie, they connect every monomial to one that differs only in whether it contains $Q_I$ or not. The same structure is obtained by letting now all the monomials $\m_j$ act upon a single field, $\f$---the `lowest component field', each $\m_j$ producing a $(1|N)$-superpartner, $(\m_j\f)$. The collection of so-obtained component fields, $\{\f,\6(\m_1\f),\6(\m_2\f),\cdots\}$, has inherited the $N$-cubical topology, admits a linear, 1--1 and off-shell action of the supersymmetry, and will be called herein the {\em free\/} supermultiplet\footnote{This supermultiplet is also referred to as `unconstrained', `Salam-Strathdee', and in contrast to subsequent constructions, `unprojected'.}.
 Owing to this, operations performed on the $N$-cube of $Q_I$-monomials have an obvious analogue on the $N$-cubical free supermultiplet, and we shan't bother with notational distinctions between these.

To preserve $N$, the number of supersymetries, the action with respect to which we intend to perform the quotient must not permute any of the $Q_I$'s. That is, all $N$ directions in $[0,1]^N\subset\IR^N$ must remain fixed. It then remains to perform simultaneous reflections about mid-points of some of the edges of the $N$-cube, the effect of which is the identification of the endpoints of those edges. This is implemented by operators of the form $(\a H^\b+\g\m_j)$, for some constants $\a,\b,\g$.

 For an operator $(\a H^\b+\g\m_j)$ to be idempotent and serve as a projection operator for a quotient, only those $\m_j$ may be used that square, upon using\eq{e(1|N)}, to an integral power of $H$. Such $\m_j$ all must be constructed from a doubly-even (divisible by 4) number of different $Q_I$'s and we denote:
\begin{equation}
 \hm_j:=Q_1^{a_{j1}}\,Q_2^{a_{j2}}\cdots Q_N^{a_{jN}}~,\quad\text{such that }
 \sum_{l=1}^N a_{jl} = 0 \bmod 4~. \label{e0mod4}
\end{equation}

 The projection operators\footnote{It is possible to `clear the denominator' and avoid negative powers of the Hamiltonian, $H$, by first transforming the free supermultiplet, using the vertex-raising of Ref.\cite{r6-1}, into the corresponding `base' supermultiplet\cite{rA}, where all bosons and all fermions, respectively, have the same engineering dimension. This effectively replaces $H$ in the present discussion with the identity operator.}
\begin{equation}
 \P_j:=\big(H^{\b_j}\pm\hm_j\big)/(2H^{\b_j})~,\qquad\text{where}\quad
 \b_j=\deg(\hm_j)
 \label{e}
\end{equation}
impose on the $N$-cube and the associated free supermultiplet relations of the form
\begin{equation}
 H^2\simeq Q_{I_1}Q_{I_2}Q_{I_3}Q_{I_4}~,\qquad\text{for some different }I_1,I_2,I_3,I_4~,
 \label{gen}
\end{equation}
where the monomial on the right-hand side of\eq{gen} is doubly-even, \ie, its degree is divisible by four.
 Such a relation implies a whole host of others, obtained by multiplying through with one or more of the $Q_I$'s; for example:
\begin{alignat}{3}
 \text{\Eq{gen}}\8 Q_{I_5}&:\quad&
 H^2\,Q_{I_5} &\simeq Q_{I_1}Q_{I_2}Q_{I_3}Q_{I_4}Q_{I_5}~;\\[2mm]
 \text{\Eq{gen}}\8 Q_{I_4}&:\quad&
 H^2\,Q_{I_4} &\simeq Q_{I_1}Q_{I_2}Q_{I_3}Q_{I_4}Q_{I_4}
               = Q_{I_1}Q_{I_2}Q_{I_3}\,H~,\nn\\[-1mm]
 \text{so}&&H\,Q_{I_4}
            &\simeq Q_{I_1}Q_{I_2}Q_{I_3}~;\\[2mm]
 \text{\Eq{gen}}\8 Q_{I_3}&:\quad&
 H^2\,Q_{I_3} &\simeq Q_{I_1}Q_{I_2}Q_{I_3}Q_{I_4}Q_{I_3}
               = -Q_{I_1}Q_{I_2}\,H\,Q_{I_4}~,\nn\\[-1mm]
 \text{so}&&H\,Q_{I_3}
            &\simeq -Q_{I_1}Q_{I_2}Q_{I_4}~;\\[2mm]
 \text{\Eq{gen}}\8 Q_{I_4}Q_{I_3}&:\quad&
 H^2\,Q_{I_4}Q_{I_3} &\simeq Q_{I_1}Q_{I_2}Q_{I_3}Q_{I_4}Q_{I_4}Q_{I_3}
               = Q_{I_1}Q_{I_2}\,H^2~,\nn\\[-1mm]
 \text{so}&&Q_{I_4}Q_{I_3}
            &\simeq Q_{I_1}Q_{I_2}~;
\end{alignat}
and so on. This operation may be pictured as the identification of all points mapped into each other through a simultaneous reflection through the mid-points of all $I_1$,- $I_2$,- $I_3$- and $I_4$-edges.

An equation such as\eq{gen} then clearly projects on the zero-set (kernel) of $\P_j:=(H^2-\hm_j)/(2H^2)$. It is not hard to show that $\P_j$ are all idempotent, and $\P_j+\bar{\P}_j=\Ione$, for every $j$, if $\bar{\P}_j:=(H^2+\hm_j)/(2H^2)$.

The reflection map $H^2\iff\hm_j$, with respect to which an equation of the form\eq{gen} describes the quotient will also be called $\hm_j$; each such map is an {\em involution\/} of the $N$-cube, \ie, a $\ZZ_2$-reflection. Consequently, division by the action of any $\hm_j$ halves the number of nodes.

For an iterated quotient, say $(\6(I^N/\hm_i)/\hm_j)$, it is necessary and sufficient that $[\hm_i,\hm_j]=0$. This implies that $\hm_i$ and $\hm_j$ must have an even number (possibly none) of $Q_1,\cdots,Q_N$ in common.

It then remains ``merely'' to find all such possible inequivalent combinations of such $\hm_j$-projections, counting two combinations equivalent if one can be turned into the other by a permutation of the $Q_1,\cdots,Q_N$ and of the $\hm_1,\cdots,\hm_k$.

\subsubsection{Listing the Quotients}
 \label{s:1List8}
The first few cases are easy to list (for complete proofs, see Ref.\cite{r6-3}):
\begin{itemize}
 \item For $N<4$, no projection can exist|we need at least four $Q_I$'s. Thus, for each of $N=1,2,3$, we have a single topology:\\
 \leavevmode\hbox to2pc{\hss\boldmath$*$~}%
  $I^1$ (interval), $I^2$ (square) and $I^3$ (cube).
 \item For $N=4$, there is only one projection $4n$-omial, $\hm_1=Q_1Q_2Q_3Q_4=:\m_{[1234]}$, which identifies the antipodal nodes in a 4-cube. The topology of the quotient, $I^4/\hm_1$, can be depicted easily, by identifying the antipodal nodes in the Adinkra with the $I^4$ (the hypercube) topology, as illustrated here:
\begin{equation}
 \vC{\includegraphics[height=40mm]{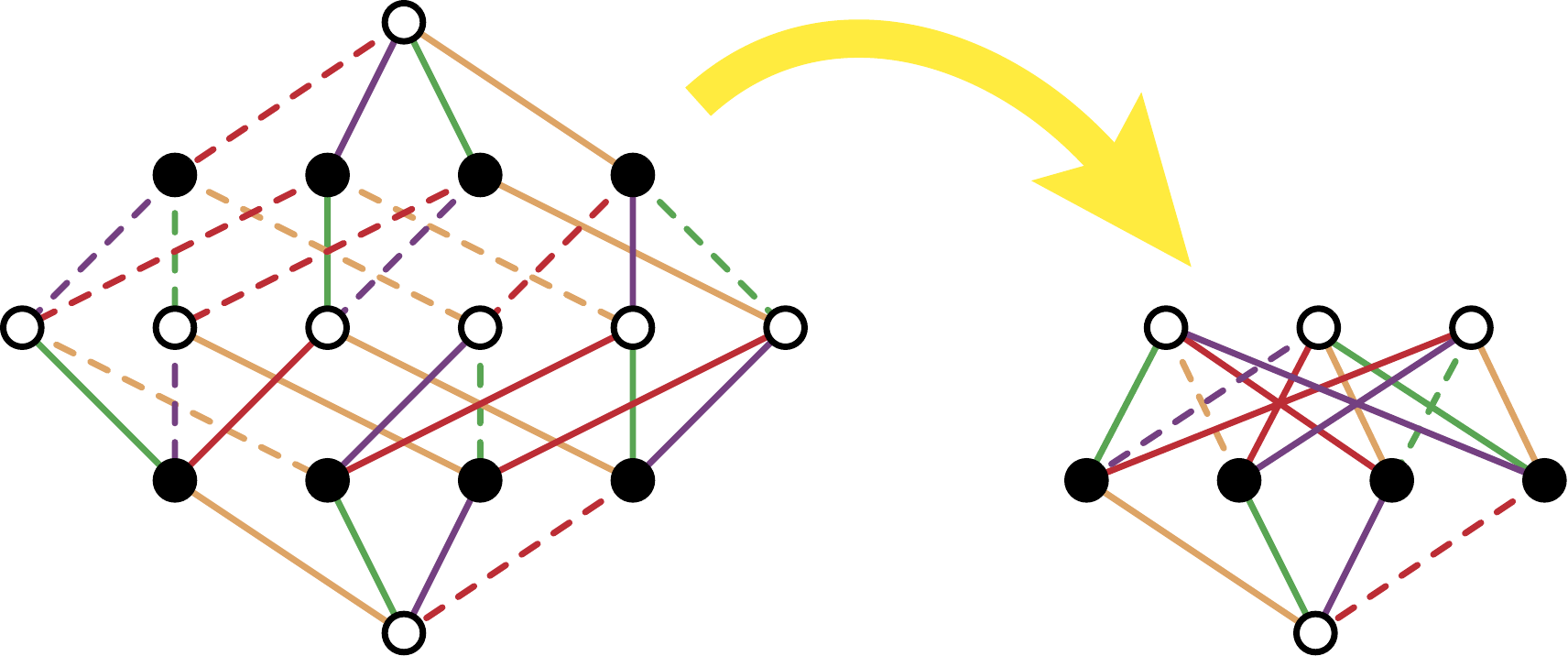}} \label{eB14641}
\end{equation}
Thus, for $N=4$, we have two distinct topologies:\\
 \leavevmode\hbox to2pc{\hss\boldmath$*$~}%
 $I^4$ and $D_4:=(I^4/\hm_1)$, depicted respectively to the left and to the right in\eq{eB14641}.
 \item For $N=5$, there are ${5\choose4}=5$ choices of which four $Q_I$'s to include in $\hm_j$, but these are all equivalent by $Q_I$-permutations. All of these leave one of the five lattice directions in the 5-cube intact, so that the result of this quotient must be a direct product of the $N=4$ quotient and the fifth, unprojected direction. Thus, there exist two $N=5$ topologies:\\
 \leavevmode\hbox to2pc{\hss\boldmath$*$~}%
 $I^5$ and $I^1\8 D_4$.
 \item For $N=6$, there are two mutually commuting $\hm_j$'s; up to a $Q_I$-permutation, these are $\hm_1=\m_{[1234]}$ and $\hm_2=\m_{[1256]}$. Notice that their composition, $\hm_1\circ\hm_2=\m_{[3456]}$, also satisfies the above requirements, which we may call $\hm_{1+2}$. It is not hard to see that $(I^6/\hm_1)\simeq(I^6/\hm_2)\simeq(I^6/\hm_{1+2})\simeq(I^2\8 D_4)$, where `$\simeq$' denotes equivalence by $Q_I$- and $\hm_j$-permutations. However, we can now also form a double quotient, $(I^6/\hm_i)/\hm_j$, where $\hm_i,\hm_j$ are any two of $\hm_1,\hm_2,\hm_{1+2}$; they all give equivalent results. Note that, on any such a double quotient, the third $\hm_j$ acts trivially, so that there is no triple quotient. We thus have the following three inequivalent topologies:\\
 \leavevmode\hbox to2pc{\hss\boldmath$*$~}%
$I^6$, $(I^2\8 D_4)$ and $D_6:=(\6(I^6/\hm_1)/\hm_2)$.
 \item For $N=7$, there are three mutually commuting reflection operators, up to a $Q_I$-permutation:
 $\hm_1=\m_{[1234]}$,
 $\hm_2=\m_{[1256]}$ and
 $\hm_3=\m_{[1357]}$. Their compositions,
\begin{equation}
 \hm_1\circ\hm_2=\m_{[3456]}~,\quad
 \hm_1\circ\hm_3=\m_{[2457]}~,\qquad
 \hm_2\circ\hm_3=\m_{[2367]}~,\qquad
 \hm_1\circ\hm_2\circ\hm_3=\m_{[1467]}~,
 \label{eV7}
\end{equation}
also satisfy the above requirements. Extending thus the previous procedure, we now have the following four inequivalent topologies:\\
 \leavevmode\hbox to2pc{\hss\boldmath$*$~}%
$I^7$, $(I^3\8 D_4)$, $(I^1\8 D_6)$ and $E_7:=\big((\6(I^6/\hm_1)/\hm_2)/\hm_3\big)$.
 \item For $N=8$, there are four mutually commuting $\hm_j$'s, up to a $Q_I$-permutation:
 $\hm_1,\hm_2,\hm_3$ from the previous cases, plus
 $\hm_4=\m_{[2468]}$. Their compositions\eq{eV7} plus
\begin{equation}
 \begin{gathered}
  \hm_1\circ\hm_4=\m_{[1368]}~,\quad
  \hm_2\circ\hm_4=\m_{[1458]}~,\qquad
  \hm_3\circ\hm_4=\m_{[12345678]}~,\\
  \hm_1\circ\hm_2\circ\hm_4=\m_{[2358]}~,\qquad
  \hm_1\circ\hm_3\circ\hm_4=\m_{[5678]}~,\qquad
  \hm_2\circ\hm_3\circ\hm_4=\m_{[3478]}~,\\
  \hm_1\circ\hm_2\circ\hm_3\circ\hm_4=\m_{[1278]}~,
 \end{gathered}
 \label{eE8}
\end{equation}
also satisfy the above requirements. Extending thus the previous procedure, we now have the following seven inequivalent topologies:\\
 \leavevmode\hbox to2pc{\hss\boldmath$*$~}%
$I^8$, $A_8:=(I^8/\m_{[12345678]})$, $(I^4\8 D_4)$, $(D_4\8 D_4)$, $(I^2\8 D_6)$, $D_8$, $(I^1\8 E_7)$ and, finally,\newline
\hbox to2pc{\hss}$E_8:=\big(\big((\6(I^6/\hm_1)/\hm_2)/\hm_3\big)/\hm_4\big)$.\\
Notice the `new' construction, $(D_4\8 D_4)$, which uses the facts that:
\begin{enumerate}
 \item $I^8=I^4_{1234}\8 I^4_{5678}$---the 8-cube is the direct product of the 4-cube generated by $Q_1,\cdots,Q_4$ and $I^4_{5678}$ generated by $Q_5,\cdots,Q_8$;
 \item $\hm_1$ operates exclusively within $I^4_{1234}\subset I^8$, and the combined reflection $(\hm_1\circ\hm_3\circ\hm_8)$ exclusively within $I^4_{5678}\subset I^8$.
\end{enumerate}
Thus, $D_4\8 D_4:=(I^4/\hm_1)\8(I^4/(\hm_1\circ\hm_3\circ\hm_8))$.
\end{itemize}
The notation ``$(D_4\8 D_4)$'' does not preclude monomials that contain $Q_I$'s from among both $Q_1,\cdots,Q_4$ and $Q_5,\cdots,Q_8$. Instead, it means that it is possible to {\em generate\/} the entire collection of these monomials from the two disjoint sets, $\{\m_{[1234]}\}$ and $\{\m_{[5678]}\}$. So, in fact, the full set is $(D_4\8D_4)=\{\Ione,\m_{[1234]},\m_{[5678]},\m_{[12345678]}\}$, and indeed may also be generated by the non-disjoint $\m_{[5678]}$ and $\m_{[12345678]}$, say. While this does not obstruct the identification of the set as $(D_4\8D_4)$, since it {\em can\/} be generated from the disjoint $\{\m_{[1234]}\}$ and $\{\m_{[5678]}\}$, it does point to an ambiguity that makes is difficult to identify mutually commuting monomial sets each given in terms of a generating set only.

Although we have examined only the few lowest-$N$ cases, it is obvious that there is an upper limit on how many mutually commuting $4n$-omials can be found, constructed from a fixed number,  $N$, of $Q_I$'s. We give here a recursive definition of this function:
\begin{equation}
 \vk(N):=
  \begin{cases}
   0 &\text{for $0\leq N<4$};\\
   \big\lfloor\frac{(N-4)^2}{4}\big\rfloor+1 &\text{for $N=4,5,6,7$};\\
   \vk(N{-}8)+4 &\text{for $N>7$, recursively}.
  \end{cases}
 \label{eKmax}
\end{equation}
and refer to the literature\cite{rPGMass,r6-3} for proofs.

Finally, the topologies obtained by using $\hm_j$-projection operators from two mutually commuting sets that are related by $Q_I$- and $\hm_j$-permutations are indistinguishable.
 We thus proceed to classify the different types of such iterated $\ZZ_2$-quotients up to $Q_I$- and $\hm_j$-permutations.

\section{Reformulations of the Problem}
 \label{s:RP}
The recursive procedure started in section~\ref{s:AT} is fairly clear and intuitive. It may be summarized as the following algorithm:
\begin{enumerate}
 \item\label{one} Find all inequivalent maximal sets, $\mathscr{S}_\a$, of mutually commuting doubly-even $Q_I$-monomials,
 \item\label{two} For each $\sS_\a$, construct the power-set, $P(\sS_\a)$, of all products of $Q_I$-monomials within $\sS_\a$, then identify those elements of $P(\sS_\a)$ which are equivalent up to $Q_I$- and $\hm_j$-permutations; call this set $P_*(\sS_\a)$.
 \item Construct all iterated quotients of the form $(I^N/\hm_{\a,j})$, for $\hm_{\a,j}\in P_*(\sS_\a)$.
\end{enumerate}
However, explicit constructions demonstrate that this becomes humanly intractable for larger $N$, and more efficient approaches are desired\cite{r6-3}.

The first objective is to efficiently address the task~\ref{one} in the above algorithm. Thereafter, the task~\ref{two} is straightforward, but obviously computation-intensive: it involves first the construction of the power-set, $P(\sS_{\a})$, for each set $\sS_{\a}$ and then the sifting for $Q_I$- and $\hm_j$-permutation inequivalent elements to obtain the set of permutation-inequivalent subsets, $P_{*}(\sS_{\a})$; a pre-selective construction of the latter would be clearly preferable.

\subsection{A Graphic Depiction}
 \label{s:GD}
Since the supersymmetry generators, $Q_{I}$, may be freely permuted, one may think of them as an unordered set, $V$, represented graphically by $N$ nodes. The desired $Q_I$-monomials\eq{e0mod4} then correspond to $w_i$-gons, $\4i{w_i}\subset V$, for which \Eq{e0mod4} implies
\begin{equation}
 \hm_i\mapsto\4i{w_i}~,\qquad |\4i{w_i}|=w_i=0\bmod4 \label{e0Mod4}
\end{equation}
and are representable as {\em $w_i$-gons\/} within $V$; to distinguish the $\4i{w_i}$'s, we may color them distinctly. 

The function
\begin{equation}
  d(\4i{w_i},\4j{w_j}):=w_i+w_j-2\big|\4i{w_i}\cap\4j{w_j}\big|~, \label{eHDij}
\end{equation}
counts in how many nodes $\4i{w_i}$ and $\4j{w_j}$ differ. Corresponding to the identity $\Ione=\m_{[-]}$ (the empty $Q_I$ monomial, with no $Q_I$'s), we denote the empty graph (no nodes and no edges) by `{\bf0}'. We then have that
\begin{equation}
 d(\4i{w_i},{\bf0}) = w_i~,\qquad\forall i\quad\&\quad0\leq w_i\leq|V|~. \label{eHDg}
\end{equation}
The condition that $\hm_i$ and $\hm_j$ have an even number of $Q_I$'s in common implies that
\begin{equation}
 |\4i{w_i}\cap\4j{w_j}|\id0\bmod2~. \label{eEven}
\end{equation}
 Eqs.\eq{e0Mod4} and\eq{eEven} then imply that:
\begin{equation}
 d(\4i{w_i},\4j{w_j})~\id~0\bmod4~. \label{eHDij2E}
\end{equation}
Thus, by iteratively listing $\4i{w_i}$'s subject to conditions\eq{e0Mod4} and~\eq{eEven}, the condition\eq{eHDij2E} is guaranteed for every pair in such a list. Furthermore, writing:
\begin{align}
 \4i{w_i}+\4j{w_j}&:= (\4i{w_i}\cup\4j{w_j}) \setminus (\4i{w_i}\cap\4j{w_j})~,
 \label{eSum}\\[-6mm]
\I{we have that:}
 \4i{w_i}+\4j{w_j}&=\4k{w_k}~,\quad\text{with}\quad
   w_k=d(\4i{w_i},\4j{w_j})\buildrel{\eq{eHDij2E}}\over\id0\bmod4~. \label{eSumW}
\end{align}
That is, the sum of any two of our $4n$-gons is also a $4n$-gon.
From this, it also follows that:
\begin{alignat}{5}
 \text{If }\4i{w_i}\cap\4j{w_j}&=\emptyset~, \quad&&\text{then}&\quad
 w_i,w_j&<|\4i{w_i}+\4j{w_j}|~; \label{eDisCon}\\[2mm]
 \text{If }\4i{w_i}&\subset\4j{w_j}~, \quad&&\text{then}&\quad
 \4i{w_i}\cap\4j{w_j}&=\4i{w_i}~,\quad\text{and} \nn\\[-1mm]
&&&&\quad
 \4i{w_i}+\4j{w_j}&=\4j{w_j}\setminus\4i{w_i}~. \label{eInclud}
\end{alignat}
The results\eqs{eDisCon}{eInclud} permit us to systematically select the smaller $w_i$-gons: For $\4i{w_i}+\4j{w_j}$ to be smaller than $\4j{w_j}$ and so its more efficient replacement, Eqs.\eq{eSumW} and\eq{eHDij} would have to imply that
\begin{equation}
 \big|\4i{w_i}+\4j{w_j}\big| < g_j~,\qquad\text{\ie}\qquad
 \big|\4i{w_i}\cap\4j{w_j}\big| > \inv2g_i~. \label{eIneq1}
\end{equation}
Now,
\begin{enumerate}
 \item If $w_i=4$, then\eq{eIneq1} implies that $|\4i{w_i}\cap\4j{w_j}|>2$, \ie, $|\4i{w_i}\cap\4j{w_j}|\geq4$ whereupon $\4i{w_i}\subset\4j{w_j}$, so that $\4j{w_j}$ may be replaced by the smaller $\4i{w_i}+\4j{w_j}$, by \Eq{eInclud}.
 \item If $w_i=8$, then\eq{eIneq1} implies that $|\4i{w_i}\cap\4j{w_j}|>4$, and owing to\eq{eSumW} then $|\4i{w_i}\cap\4j{w_j}|\geq8$; but then $\4i{w_i}\subset\4j{w_j}$, so that $\4j{w_j}$ may be replaced by the smaller $\4i{w_i}+\4j{w_j}$, by \Eq{eInclud}.
 \item If $w_i=12$, then\eq{eIneq1} implies that $|\4i{w_i}\cap\4j{w_j}|>6$, \ie, $|\4i{w_i}\cap\4j{w_j}|\geq8$; it then remains undetermined which are the smaller two from among $\4i{w_i},\4j{w_j}$ and $\4i{w_i}+\4j{w_j}$.
\end{enumerate}
This proves:
\begin{prop}
Let $\4i{w_i},\4j{w_j}$ both be $w_*$-gons satisfying the conditions\eq{e0Mod4} and\eq{eEven}.
Then, for $w_i\leq8$, the condition\eq{eIneq1} is both necessary and sufficient for replacing $\4j{w_j}$ with the smaller $\4i{w_i}+\4j{w_j}$; for $w_i\geq12$, \Eq{eIneq1} is only necessary.
\end{prop}

\subsection{Error-Detecting/Correcting Codes}
 \label{s:ED/CC}
We note that every monomial constructed from the $N$ supersymmetry generators $Q_I$ corresponds, in a 1--1 fashion, to a binary $N$-digit number ($N$-vector), the digits (components) of which are read off as follows:
\begin{equation}
 \hm_j=Q_1^{~a_{j1}}\,Q_2^{~a_{j2}}\,\cdots\,Q_N^{~a_{jN}}~\mapsto~
 \ba_j:=(a_{j1},a_{j2},\cdots,a_{jN})~,\quad a_{jl}\in\{0,1\}~. \label{eQ2B}
\end{equation}
and where
\begin{equation}
 \sum_{l=1}^N a_{jl}=w_j\id0\bmod4~,\qquad j=1,\cdots,k~, \tag{\ref{e0mod4}$'$}
\end{equation}
echoes Eqs.\eq{e0mod4} and~\eq{e0Mod4}.
Thus, for example:
\begin{equation}
 N{=}6:
 \left\{
  \begin{array}{ccl}
   \hm_1&\mapsto&001111~,\\[-3pt]
   \hm_2&\mapsto&111100~,\\[-3pt]
   \hm_1\circ\hm_2&\mapsto&110011~;\\
  \end{array}\right.
 \quad
 N{=}7:
 \left\{
  \begin{array}{rcl}
   \hm_1&\mapsto&0001111~,\\[-3pt]
   \hm_2&\mapsto&0111100~,\\[-3pt]
   \hm_1\circ\hm_2&\mapsto&0110011~,\\[-3pt]
   \hm_3&\mapsto&1010101~,\\[-3pt]
   \hm_1\circ\hm_3&\mapsto&1011010~,\\[-3pt]
   \hm_2\circ\hm_3&\mapsto&1101001~,\\[-3pt]
   \hm_1\circ\hm_2\circ\hm_3&\mapsto&1100110~;\\
  \end{array}\right. \label{eD6V7}
\end{equation}
and so on.

It turns out that each such collection of $k$ binary $N$-digit numbers ($N$-vectors) forms an (2-error-detecting and 1-error-correcting) $[N,k]$-code\cite{rMcWS,rCHVP}, and their double-evenness (divisibility of $w_i$ by 4) by virtue of \Eq{eSumW}, guarantees that they are all {\em self-orthogonal\/}. In fact, most of the preceding section, \SS~\ref{s:GD}, is straightforward to re-cast into the corresponding statements about such $[N,k]$-codes. In fact, the function\eq{eHDij} is well known as the {\em Hamming distance\/} between two code-words, and\eq{eHDg} as the {\em Hamming weight\/} of a code-word.

 Just as the monomials $\hm_{1}=Q_1Q_2Q_3Q_4,\hm_{2}=Q_3Q_4Q_5Q_6$ for $N=6$, and $\hm_{1}=Q_1Q_2Q_3Q_4,\hm_{2}=Q_3Q_4Q_5Q_6,\hm_{3}=Q_1Q_3Q_5Q_7$ for $N=7$, \etc, generate all the others in the collections\eq{eD6V7}, every other set, $\sS_\a$ of mutually commuting, doubly-even monomials in $Q_1,\cdots,Q_N$ has a {\em generating set\/}. These generating sets are by no means unique: the full collection of $\hm_j$-monomials for $N=6$ is equally well generated by either two from among $\hm_1,\hm_2,(\hm_1\circ\hm_2)$, and all choices are equivalent to each other via $Q_I$- and $\hm_j$-permutations. For higher $N$,
\begin{enumerate}
 \item the number of such distinct generating sets is combinatorially larger,
 \item they are not all equivalent to each other, 
 \item both the automorphism groups and the number of inequivalent classes grows combinatorially.
\end{enumerate}
 
The situation is identical with the $[N,k]$-codes specified by the mapping\eq{eQ2B}. 

There is a considerable body of literature, even in book form\cite{rMcWS,rCHVP}, on error-detecting/correcting codes, but much of the research in this field focuses on codes containing code-words with weights no smaller than a certain minimum. By contrast, the very definition\eq{e0Mod4} implies that we are interested in codes consisting only of doubly-even code-words, \ie, the Hamming weight of every code-word must be divisible by 4, not merely `larger than 3'.

Furthermore, $(1|N)$-supersymmetry and its representations occur in physics applications in at least three different ways: ({\bf1})~dimensional reduction of the ({\em real\/}) spacetime in field theories, ({\bf2})~dimensional reduction of the ``underlying''  theory\footnote{This is well understood for (super)string theory, where the ``underlying'' theory is a quantum field theory with the world-sheet of the string as the domain space; it is not known precisely what sort of theories constitute the ``underlying'' theories of $M$- and $F$-theory.} of string-theory and its $M$- and $F$-theory extensions, and ({\bf3})~effective supersymmetry in the Hilbert space of all supersymmetric quantum field theories. In case~({\bf1}), $N$ will be limited, except for the low, $N=1,2,3$ cases, to being divisible by 4. However, in the remaining two cases, this restriction is far less strict and depends on the details of the dynamics of the concrete model considered. However, various considerations place a relatively well-accepted {\em upper\/} bound: $N\leq32$ from considerations in $M$-theory\cite{rJPS}.

Given a fixed $N$, it is desirable to find {\em all\/} indecomposable representations of $(1|N)$-supersym\-metry, as all these may, conceivably, serve as building blocks in various such theories. This implies that we {\em do\/} need all possible Adinkra topologies, \ie, all $[N,k]$-codes, where $k\leq\vk(N)$ as defined in \Eq{eKmax}.

Binary $[N,\vk\6(N)]$-codes are {\em maximal\/} in the information-theoretic sense guaranteed by Shannon's theorem: the fraction $k/N$ is called the {\em information rate\/} of an $[N,k]$-code and measures how much of the information is being transmitted. The maximal $\vk(N)$ therefore corresponds to the maximal information rate: $\vk\6(N)/N\leq50\,\%$, and $\vk\6(N)/N=50\,\%$ only when $N=0\bmod8$.

In turn, for any fixed $N$, the $[N,\vk\6(N)]$-codes also correspond to {\em maximal\/} iterated $\ZZ_2$-quotients of $I^N$; that is, the iterated quotient is done the maximal, $\vk(N)$, number of times. Dually, the number of nodes in the resulting Adinkra is {\em minimal\/} for the fixed $N$: the size of the corresponding indecomposable $(1|N)$-supermultiplet, $(2^{N-k-1},2^{N-k-1})$, is minimal for $k=\vk(N)$.

Note, however, that these codes and Adinkra topologies for $N>9$ are {\em not unique\/}! For example, rather famously, there exist two distinct maximal binary doubly-even $[16,8]$-codes, corresponding to the Lie algebras $E_8\8E_8$ and $D_{16}$ 

Now, every $[N,k]$-code has a generating set consisting of $k$ $N$-digit binary numbers, which corresponds to a binary matrix of $k$ rows and $N$ columns. This matrix is of maximal rank\cite{rCHVP}. Using Gauss-Jordan row operations and column-permutations only (and so replacing a generating set by an equivalent one), this matrix can then be brought into the format $[\,\Ione_k\,|\,\IA\,]$, where $\IA$ is an $(N{-}k)\times k$ binary matrix, called the {\em redundancy part\/}, and $\Ione_k$ is the $k\times k$ identity matrix.

It is then tempting to ({\em falsely\/}) conclude that all binary $[N,k]$-codes are classified by classifying all smaller $(N{-}k)\times k$ matrices, $\IA$.

On one hand, the list of all inequivalent binary $(N{-}k)\times k$ matrices will be redundant as a list of all binary $[N,k]$-codes, since the complete generator set matrix, $[\,\Ione_k\,|\,\IA\,]$ admits more column-permutations than $\IA$ alone: inequivalent $\IA$-matrices may correspond to equivalent $[\,\Ione_k\,|\,\IA\,]$-matrices. With this problem alone, we could still classify all inequivalent $\IA$-matrices, and thereafter sort out the `external' equivalences due to the larger automorphism group of $[\,\Ione_k\,|\,\IA\,]$.

On the other hand, however, a list of inequivalent binary $\IA$-matrices also provides an {\em incomplete\/} list of binary $[N,k]$-codes:
\begin{prop}
The redundancy parts of generator sets of binary $[N,k]$-codes are not all sub-matrices of the redundancy parts of the generator sets of {\em maximal\/} binary $[N,\vk\6(N)]$-codes.
\end{prop}
\begin{proof}
Since the redundancy part has $(N{-}k)$ columns, every binary $[N,k]$-code can be generated by code-words of Hamming weight no more than $(N{-}k){+}1$: at most $(N{-}k)$ 1' in the redundancy part, and precisely one 1 in $\Ione_k$.

But, for $k<\vk(N)$, this means that a binary $[N,k]$-code may well have generators of weight $(N{-}k)>(N{-}\vk\6(N))$, and so not be equivalent to any subcode of any binary $[N,\vk(N)]$-code.
\end{proof}
\Remk
An extreme example is provided as follows: binary $[32,16]$-codes are maximal in that they are generated by the maximal number of generators, $16=\vk(32)$. The redundancy part of such codes' generator sets is a $16\time16$ binary matrix. The maximal weight of the generators is then no more than $16+1$, and having to be divisible by four, it must be 16. The maximal weight of the redundancy part of any generator must therefore be 15.

By contrast, the generator $1111\,1111\,1111\,1111\,1111\,1111\,1111\,1111$ of the binary $[32,1]$-code has weight 32. The weight of its redundancy part is 31, and so could not be a sub-matrix of the redundancy part of any binary $[32,16]$-code.

\Remk
Of course, the {\em full\/} binary $[32,16]$-codes may well (in fact, they always do) contain the above weight-32 codeword, so that another generator set of {\em any\/} binary $[N,k]$-code would have contained this code-word.

In fact, we have:
\begin{prop}
For every $k<\vk(N)$, every binary $[N,k]$-code is a subcode of some binary $[N,\vk\6(N)]$-code.
\end{prop}
\begin{proof}
It follows from the proofs of Theorem~7 and Lemma~7 of Ref.\cite{rPGMass}, that every binary, doubly-even $[N,k]$-code, with $k<\vk(N)$, may be completed to a binary, doubly-even $[N,\vk\6(N)]$-code by adding $(\vk(N){-}k)$ linearly independent generators and the generated code-words; see Ref.\cite{r6-3}.
\end{proof}
However, this may be difficult to discern from any preferred choice of generator sets.

\subsection{$1\leq N\leq8$}
 \label{s1N8}
The results for $1\leq N\leq8$ are presented in Fig~\ref{f1N8}.
Then, for example, the code-word corresponding to the tetragon $\414=\{1,2,3,4\}$ corresponds to the binary number $\cdots00001111$, $\424=\{1,3,4,6\}$ corresponds to $\cdots00101101$, $\434=\{1,2,5,6\}$ corresponds to $\cdots00110011$, and so on. 
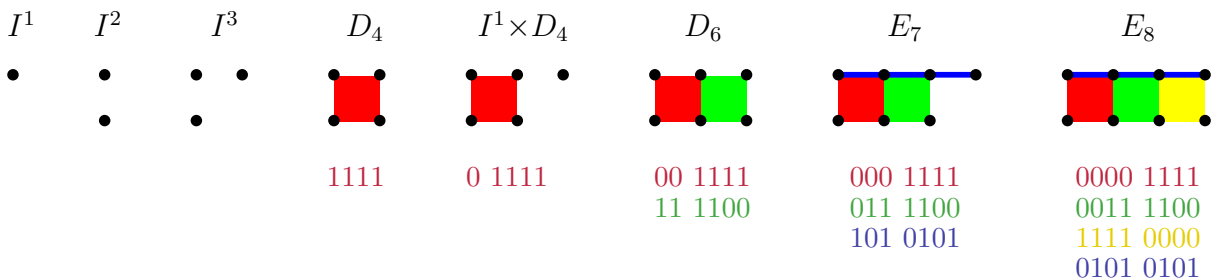
\begin{figure}[ht]
\begin{center}
 \begin{picture}(140,40)(10,0)
  \put(0,20){\includegraphics[width=160mm]{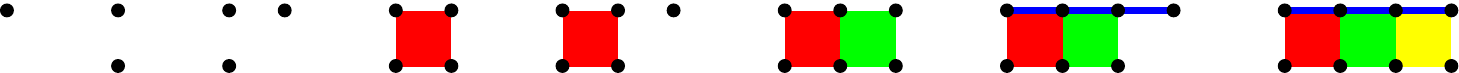}}
  \put(0,32){$I^1$}
  \put(11.5,32){$I^2$}
  \put(27,32){$I^3$}
  \put(45,32){$D_4$}
  \put(62.5,32){$I^1\8D_4$}
  \put(90,32){$D_6$}
  \put(117,32){$E_7$}
  \put(148,32){$E_8$}
  \put(42.5,12){\small$\color{Red}{1111}$}
  \put(61,12){\small$\color{Red}{0~1111}$}
  \put(86,12){\small$\color{Red}{00~1111}$}
  \put(86,8){\small$\color{Green}{11~1100}$}
  \put(112,12){\small$\color{Red}{000~1111}$}
  \put(112,8){\small$\color{Green}{011~1100}$}
  \put(112,4){\small$\color{Blue}{101~0101}$}
  \put(142,12){\small$\color{Red}{0000~1111}$}
  \put(142,8){\small$\color{Green}{0011~1100}$}
  \put(142,4){\small$\color{yel}{1111~0000}$}
  \put(142,0){\small$\color{Blue}{0101~0101}$}
 \end{picture}
\caption{The tetragons representing the binary $[N,\vk(N)]$-codes' generator sets for $1\leq N\leq8$.}
\label{f1N8}
\end{center}
\end{figure}
The codes are then read off the graphs by identifying the nodes as binary digits, following the pattern:
\begin{equation}
 \begin{matrix}
  1 & 3 & 5 & 7 & \cdots \\
  2 & 4 & 6 & 8 & \cdots
 \end{matrix} \label{ePat}
\end{equation}
Initially, the code-words themselves are given underneath the graphs, inked in a color that corresponds to the tetragons in the graphs. Note that the blue tetragon is depicted as a straight line, present in the $E_7$ and $E_8$ graphs, whereas the red, green and yellow ones are squares; this appearance is of course immaterial. There are also many other depictions of the same and equivalent colored subsets; the essential is, however, the adherence to \Eq{eEven}.

Also, for each $N$, the generator sets of the {\em maximal\/} codes are presented. This then corresponds to a {\em maximal\/} iterated $\ZZ_2$-quotient of $I^N$; `intermediate' topologies are obtained by dropping one or more of the generators, \ie, binary vectors, \ie, $\ZZ_2$ quotients from the iteration.

Combining the well known fact\cite{rCHVP} that binary $[N,k]$-codes can be obtained from binary $[N',k]$-codes by {\em puncturing\/} when $N'=N{+}1$ and the above results, we see that all binary doubly-even $[N,k]$-codes with $N\leq32$ and $k\leq\vk(N)$ may be obtained one from another, and so form a tree: one starts from the empty code, which is the case for
\begin{equation}
 N=1:~ [1,0]=[\Stk{0}]~,\qquad
 N=2:~ [2,0]=[\Stk{00}]~,\qquad
 N=3:~ [3,0]=[\Stk{000}]~,
\end{equation}
and obtains the first non-empty code at
\begin{equation}
 N=4:~ [4,0]=\left[\Stk{0000}\right] \subset
       [4,1]=\left[\Stk{0000\\ 1111}\right]~,
\end{equation}
with the empty code its subset. To conserve space, we revert to only listing a generator set rather than the full code. Adopting the notation $[\IG]^N_k$, where $\IG$ is a generator set of a code, we have the initial stages of our tree of codes, based on the results from \SS~\ref{s:1List8} and also given in Figure~\ref{f1N8}:
\begin{equation}
 \vC{\begin{picture}(140,73)(8,-23)
  \put(0,47){$[-]^1_0$}
  \put(10,48){\line(1,0){6}}
  \put(17,47){$[-]^2_0$}
  \put(27,48){\line(1,0){6}}
  \put(34,47){$[-]^3_0$}
  \put(44,48){\line(1,0){6}}
  \put(51,47){$[-]^4_0$}
  \put(61,48){\line(1,0){9}}
  \put(71,47){$[-]^5_0$}
  \put(80,48){\line(1,0){11}}
  \put(92,47){$[-]^6_0$}
  \put(101,48){\line(1,0){12}}
  \put(114,47){$[-]^7_0$}
  \put(123,48){\line(1,0){14}}
  \put(137,47){$[-]^8_0$}
  \put(123,46){\line(3,-1){10}}
  \put(133,41){$[\Stk{11111111}]^8_1$}
  \put(44,45){\line(1,-2){5}}
  \put(49,31){$[\Stk{1111}]^4_1$}
  \put(61,32){\line(1,0){6}}
  \put(68,31){$[\Stk{01111}]^5_1$}
  \put(82,32){\line(1,0){6}}
  \put(89,31){$[\Stk{001111}]^6_1$}
  \put(105,32){\line(3,1){5}}
  \put(110,33){$[\Stk{0001111}]^7_1$}
  \put(128,35){\line(1,0){5}}
  \put(133,34){$[\Stk{00001111}]^8_1$}
  \put(128,34){\line(2,-1){5}}
  \put(132.5,27){$\left[\Stk{00001111\\ 11110000}\right]^8_2$}
  \put(82,30){\line(1,-1){6.5}}
  \put(89,19){$\left[\Stk{001111\\ 111100}\right]^6_2$}
  \put(105,20){\line(1,0){4}}
  \put(109,19){$\left[\Stk{0001111\\ 0111100}\right]^7_2$}
  \put(128,20){\line(1,0){4}}
  \put(132.5,19){$\left[\Stk{00001111\\ 00111100}\right]^8_2$}
  \put(128,17){\line(1,-1){5}}
  \put(132.5,9){$\left[\Stk{00001111\\ 00111100\\ 11110000}\right]^8_3$}
  \put(105,16){\line(1,-3){5}}
  \put(109.5,-3){$\left[\Stk{0001111\\ 0111100\\ 1010101}\right]^7_3$}
  \put(128,-2){\line(1,0){4}}
  \put(132,-3){$\left[\Stk{00001111\\ 0011\,1100\\ 01010101}\right]^8_3$}
  \put(128,-7){\line(1,-1){5}}
  \put(132,-17){$\left[\Stk{00001111\\ 00111100\\ 11110000\\ 01010101}\right]^8_4$}
 \end{picture}}
 \label{eCode3}
\end{equation}
corresponding to the (evolutionary) tree of Adinkra topologies, named in \SS~\ref{s:1List8} and also in Figure~\ref{f1N8}:
\begin{equation}
 \vC{\begin{picture}(140,55)(8,-5)
  \put(3,47){$I^1$}
  \put(8,48){\line(1,0){9}}
  \put(19,47){$I^2$}
  \put(24,48){\line(1,0){9}}
  \put(35,47){$I^3$}
  \put(40,48){\line(1,0){9}}
  \put(51,47){$I^4$}
  \put(56,48){\line(1,0){13}}
  \put(71,47){$I^5$}
  \put(76,48){\line(1,0){13}}
  \put(90,47){$I^6$}
  \put(96,48){\line(1,0){13}}
  \put(112,47){$I^7$}
  \put(118,48){\line(1,0){13}}
  \put(134,47){$I^8$}
  \put(118,47){\line(4,-1){14}}
  \put(133,41){$A_8$}
  \put(40,46){\line(2,-3){9}}
  \put(50,31){$D_4$}
  \put(56,32){\line(1,0){8}}
  \put(65,31){$I^1\8D_4$}
  \put(80,32){\line(1,0){5}}
  \put(86,31){$I^2\8D_4$}
  \put(100.5,32.5){\line(3,1){5.5}}
  \put(107,33){$I^3\8D_4$}
  \put(122,35){\line(1,0){7}}
  \put(129.5,34){$I^4\8D_4$}
  \put(122,34){\line(2,-1){7}}
  \put(129.5,27){$D_4\8D_4$}
  \put(80,30){\line(4,-3){9}}
  \put(90,20){$D_6$}
  \put(96.5,21){\line(1,0){9.5}}
  \put(107,20){$I^1\8D_6$}
  \put(122,21){\line(1,0){7}}
  \put(129.5,20){$I^2\8D_6$}
  \put(122,20){\line(3,-1){11}}
  \put(133,14){$D_8$}
  \put(96,19){\line(4,-3){13}}
  \put(110,6){$E_7$}
  \put(116,7){\line(1,0){13}}
  \put(129.5,6){$I^1\8E_7$}
  \put(116,6){\line(5,-2){15}}
  \put(132.5,-2){$E_8$}
 \end{picture}}
 \label{eCode3n}
\end{equation}
and each of which, by virtue of the `hanging gardens theorem' of Ref.\cite{r6-1}, corresponds to a {\em family\/} of indecomposable representations of $(1|N)$-supersymmetry, growing combinatorially with $N$.

In ``reading'' the trees\eqs{eCode3}{eCode3n} from right to left, the links indicate ``puncturing''\cite{rCHVP}: dropping a supersymmetry generator, so $N\mapsto(N{-}1)$, whereas the left-to-right direction corresponds to ``extending'' the codes. On the other hand, each column in these, horizontal depictions of the trees\eqs{eCode3}{eCode3n}---corresponding to horizontal {\em levels\/} in the more traditional, vertical depiction, is clearly seen to be partitioned by variable $k$. A code $[\IG]^N_k$ may then be obtained from $[\IG']^N_{k+1}$ by the {\em lateral\/} transformation of dropping one of the generators, and {\em vice versa\/}, from $[\IG'']^N_{k-1}$ by adding a generator. These four operations provide the basic transformations of codes (and so Adinkra topologies) from one into another. Additional, more complex operations are also possible\cite{rCHVP}.

To conclude, we list the number of bosons and fermions in the supermultiplets with the corresponding topologies\eq{eCode3n} in Table~\ref{t1N8bf}.
\begin{table}
  \centering
 {\footnotesize
  \begin{tabular}{@{} c|c|c|c|c|c|c|c @{}}
    \toprule
    \multicolumn{8}{c}{\BM{N}}\\
    {\bf1} & {\bf2} & {\bf3} & {\bf4} & {\bf5} & {\bf6} & {\bf7} & {\bf8} \\ 
    \midrule
    \shortstack{$I^1$\\(1;1)} & \shortstack{$I^2$\\(2;2)} & \shortstack{$I^3$\\(4;4)}
  & \shortstack{$I^4$\\(8;8)} & \shortstack{$I^5$\\(16;16)} & \shortstack{$I^6$\\(32;32)}
  & \shortstack{$I^7$\\(64;64)} & \shortstack{$I^8$\\(128;128)} \\ 
    \midrule
    & & & \shortstack{$D_4$\\(4;4)} & \shortstack{$I^1\8D_4$\\(8;8)}
  & \shortstack{$I^2\8D_4$\\(16;16)} & \shortstack{$I^3\8D_4$\\(32;32)}
  & \shortstack{$I^4\8D_4$\\(64;64)} \shortstack{$A_8$\\(64;64)}\\ 
    \midrule
    & & & & & \shortstack{$D_6$\\(8;8)} & \shortstack{$I^1\8D_6$\\(16;16)}
  & \shortstack{$D_4\8D_4$\\(32;32)} \shortstack{$I^2\8D_6$\\(32;32)}\\ 
    \midrule
    & & & & & & \shortstack{$E_7$\\(8;8)} & \shortstack{$D_8$\\(16;16)}
                                             \shortstack{$I^1\8E_7$\\(16;16)}\\ 
    \midrule
    & & & & & & & \shortstack{$E_8$\\(8;8)}\\ 
    \bottomrule
  \end{tabular}}
  \begin{picture}(30,40)(-2,61)
   \put(17,80){\footnotesize $I^8$}
    \put(17,78){\vector(-1,-1){5}}
    \put(19,78){\vector(1,-1){5}}
   \put(3,68.5){\footnotesize $I^4\8D_4$}
   \put(23,68.5){\footnotesize $A_8$}
    \put(10,66.5){\vector(0,-1){5}}
    \put(23,66.5){\vector(-3,-2){8}}
    \put(25,66.5){\vector(0,-1){5}}
   \put(3,57){\footnotesize $D_4\8D_4$}
   \put(20,57){\footnotesize $I^2\8D_6$}
    \put(10,55){\vector(0,-1){6}}
    \put(23,55){\vector(-3,-2){10}}
    \put(25,55){\vector(0,-1){6}}
  \put(20,45){\footnotesize $I^1\8E_7$}
   \put(8,45){\footnotesize $D_8$}
    \put(12,43){\vector(1,-1){5}}
    \put(25,43){\vector(-1,-1){5}}
   \put(17,33.5){\footnotesize $E_8$}
  \end{picture}
  \caption{The `fundamental' indecomposable $(1|N)$-supermultiplet families for $1\leq N\leq8$, indicated by their Adinkra topology label and the number of bosons and fermions. To the right is the pattern of quotient iterations for $N=8$, the lowest case where this is not linear.}
  \label{t1N8bf}
\end{table}

\subsection{Even More Than That!}
 \label{s:More}
Recall that additional supermultiplets may be defined using the `fundamental' ones and linear algebra, in precise analogy with representations of Lie algebras; the difference is only that the number of `fundamental' indecomposable representations with which to construct other ones is increasingly larger as $N$ grows\footnote{By sharp contrast, Lie algebras have a finite number of `fundamental' representations {\em fixed per type\/}: for all $n$, $A_n$ has the complex $(n{+}1)$-vector and its conjugate, $B_n$ the (pseudo-)real $2^n$-spinor, $C_n$ the pseudo-real $2n$-vector, $D_n$ the two $2^{n-1}$-spinors, and $G_2,F_4,E_6,E_7$ and $E_8$ their 7,- 26,- 27,- 56- and 248-dimensional representations. All other representations are obtained from these as (iterated) kernels and cokernels of maps between tensor products of these `fundamental' representations.}. 

For example, let \BM{X} be a $(1|4)$-supermultiplet with the $D_4$ topology, and \BM{Y} the direct sum of two $(1|4)$-supermultiplets with the $I^4$ topology. Using the hanging garden theorem of Ref.\cite{r6-1}, the bosons and the fermions may be redefined so that, in order of increasing engineering dimension, their numbers are: $(2;4;2)$ for \BM{X}, and $(2;8;12;8;2)$ for \BM{Y}. It is now possible to define the $(1|4)$-supermultiplet $\BM{Y}/\BM{X}$ with $(0;4;10;8;2)$ independent component fields, having gauged away the lowest $(2;4;2)$ component fields of \BM{Y} by \BM{X}. This is the cokernel of a supersymmetry-preserving map $\BM{X}\to\BM{Y}$.

Alternatively, we may shift the engineering dimensions of \BM{X}, so its top component is on par with the top component of \BM{Y}; denote this $\BM{X}'$. We may then constrain \BM{Y} to the kernel of a supersymmetry-preserving map $\BM{Y}\to\BM{X}'$, which has $(2;8;10;4;0)$ independent component fields. This in fact fits the usual definition of a complex-linear superfield from $(4|1)$-supersymmetry, where the map is provided by the quadratic superderivative $\bar{D}_{\dot\a}\bar{D}^{\dot\a}$: a complex-linear supermultiplet $\BM{\J}\subset\BM{Y}$ is defined to satisfy $\bar{D}_{\dot\a}\bar{D}^{\dot\a}\BM{\J}=0$, this being the kernel of $\bar{D}_{\dot\a}\bar{D}^{\dot\a}\BM{Y}\to\BM{X}'$.

Quite clearly, this provides for a recursive explosion of $(1|N)$-representations, foreshadowed in Ref.\cite{r6-1}, which presents a semi-infinite sequence of such, ever-larger representations of the lowest non-trivial $(1|3)$-supersymmetry.

\subsection{$9\leq N\leq16$}
 \label{s9N16}
Without further ado, we find:
\begin{equation}
 \vC{\begin{picture}(140,62)(15,-5)
  \put(-1,1){\includegraphics[width=160mm]{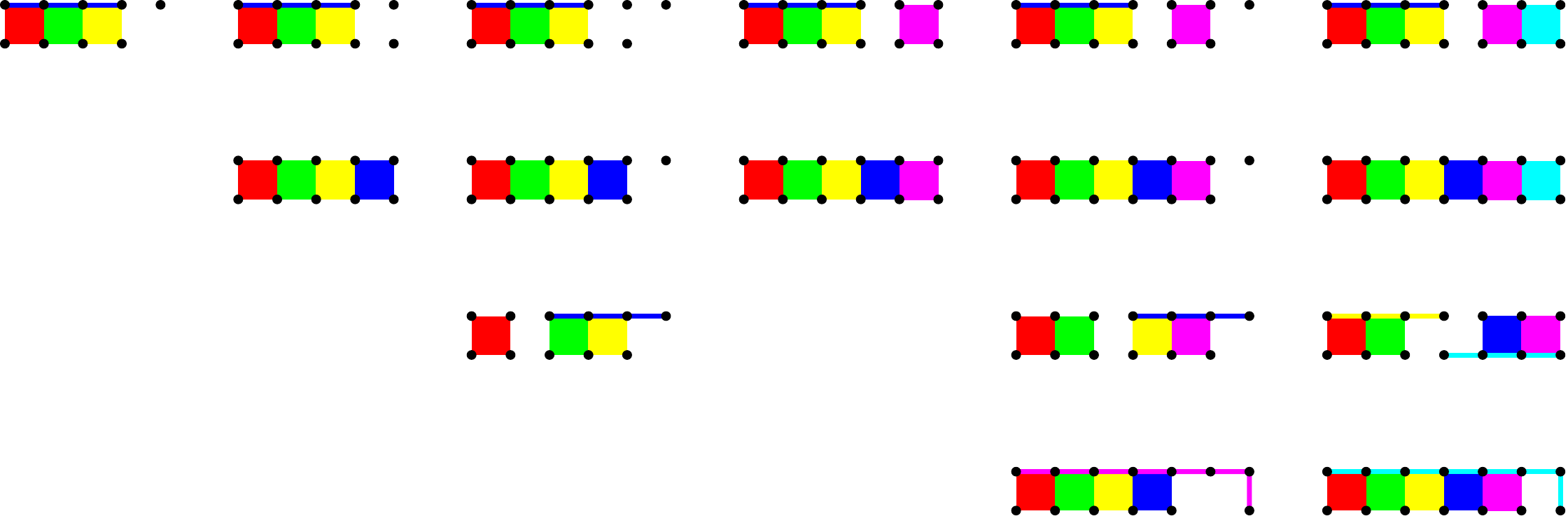}}
  \put(1,55){$N=9$}
   \put(1,43.5){\footnotesize$E_8\8I^1$}
  \put(25,55){$N=10$}
   \put(26,43.5){\footnotesize$E_8\8I^2$}
   \put(29,28){\footnotesize$D_{10}$}
  \put(51,55){$N=11$}
   \put(51,43.5){\footnotesize$E_8\8I^3$}
   \put(51,28){\footnotesize$D_{10}\8I^1$}
   \put(51,12.5){\footnotesize$D_4\8E_7$}
  \put(79,55){$N=12$}
   \put(80,43.5){\footnotesize$E_8\8D_4$}
   \put(83,28){\footnotesize$D_{12}$}
  \put(109,55){$N=13$}
   \put(108,43.5){\footnotesize$E_8\8D_4\8I^1$}
   \put(108,28){\footnotesize$D_{12}\8I^1$}
   \put(108,12.5){\footnotesize$D_6\8E_7$}
   \put(112,-3.5){\footnotesize$E_{13}$}
  \put(141,55){$N=14$}
   \put(141,43.5){\footnotesize$E_8\8D_6$}
   \put(144,28){\footnotesize$D_{14}$}
   \put(141,12.5){\footnotesize$E_7\8E_7$}
   \put(144,-3.5){\footnotesize$E_{14}$}
 \end{picture}}
 \label{e9N14}
\end{equation}
for $N\in[9,14]$, each column with a fixed, indicated $N$, and which correspond to the codes indicated in Table~\ref{t9N14}.
\begin{table}[ht]
  \centering
 {\footnotesize
  \begin{tabular}{@{} c|c|c|c|c|c @{}}
    \toprule
 \BM{N=9} & \BM{N=10} & \BM{N=11} & \BM{N=12} & \BM{N=13} & \BM{N=14} \\
    \midrule
   \stk{$I^1\8E_8$\\
        0~00001111\\
        0~00111100\\
        0~11110000\\
        0~01010101}
  &\stk{$I^2\8E_8$\\
        00~00001111\\
        00~00111100\\
        00~11110000\\
        00~01010101}
  &\stk{$I^3\8E_8$\\
        000~00001111\\
        000~00111100\\
        000~11110000\\
        000~01010101}
  &\stk{$D_4\8E_8$\\
        0000~00001111\\
        0000~00111100\\
        0000~11110000\\
        0000~01010101\\
        1111~00000000}
  &\stk{$I^1\8D_4\8E_8$\\
        00000~00001111\\
        00000~00111100\\
        00000~11110000\\
        00000~01010101\\
        01111~00000000}
  &\stk{$D_6\8E_8$\\
        000000~00001111\\
        000000~00111100\\
        000000~11110000\\
        000000~01010101\\
        001111~00000000\\
        111100~00000000} \\
    \midrule
  &\stk{$D_{10}$\\
        0000001111\\
        0000111100\\
        0011110000\\
        1111000000}
  &\stk{$I^1\8D_{10}$\\
        0~0000001111\\
        0~0000111100\\
        0~0011110000\\
        0~1111000000}
  &\stk{$D_{12}$\\
        000000001111\\
        000000111100\\
        000011110000\\
        001111000000\\
        111100000000}
  &\stk{$I^1\8D_{12}$\\
        0~000000001111\\
        0~000000111100\\
        0~000011110000\\
        0~001111000000\\
        0~111100000000}
  &\stk{$D_{14}$\\
        00000000001111\\
        00000000111100\\
        00000011110000\\
        00001111000000\\
        00111100000000\\
        11110000000000} \\
    \midrule
  &
  &\stk{$E_7\8D_4$\\
        0000000~1111\\
        0001111~0000\\
        0111100~0000\\
        1010101~0000}
  &
  &\stk{$E_7\8D_6$\\
        0000000~001111\\
        0000000~111100\\
        0001111~000000\\
        0111100~000000\\
        1010101~000000}
  &\stk{$E_7\8E_7$\\
        0000000~0001111\\
        0000000~0111100\\
        0000000~1010101\\
        0011110~0000000\\
        1111000~0000000\\
        1010101~0000000} \\
    \midrule
  &
  &
  &
  &\stk{$E_{13}$\\
        0000000001111\\
        0000000111100\\
        0000011110000\\
        0001111000000\\
        1110101010101}
  &\stk{$E_{14}$\\
        00000000001111\\
        00000000111100\\
        00000011110000\\
        00001111000000\\
        00111100000000\\
        11010101010101} \\
    \bottomrule
  \end{tabular}}
  \caption{The code-bases corresponding to the graphs\eq{e9N14}, headed with the names of the corresponding Cliffordinkras.}
  \label{t9N14}
\end{table}

Finally, for $N=15,16$, we find:
\begin{equation}
 \vC{\begin{picture}(140,48)(0,2)
  \put(2,10){\includegraphics[width=140mm]{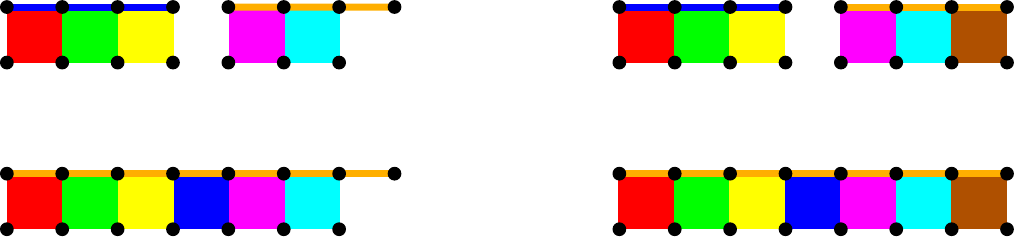}}
  \put(22,46){$N=15$}
   \put(22,28){\small$E_8\8E_7$}
   \put(25,4){\small$E_{15}$}
  \put(105,46){$N=16$}
   \put(106,28){\small$E_8\8E_8$}
   \put(110,4){\small$D_{16}$}
 \end{picture}}
 \label{e15N16}
\end{equation}
which correspond to the codes in Table~\ref{t15N16}.
\begin{table}[ht]
  \centering
  \begin{tabular}{@{} ccc|ccc @{}}
    \toprule
    \multicolumn{2}{c}{\BM{N=15}} &&& \multicolumn{2}{c}{\BM{N=16}} \\
    \midrule
   \stk{$E_7\8E_8$\\
        0000000~00001111\\
        0000000~00111100\\
        0000000~11110000\\
        0000000~01010101\\
        0001111~00000000\\
        0111100~00000000\\
        1010101~00000000}
  &\stk{$E_{15}$\\
        000000000001111\\
        000000000111100\\
        000000011110000\\
        000001111000000\\
        000111100000000\\
        011110000000000\\
        101010101010101}
 &&&\stk{$E_8\8E_8$\\
        00000000~00001111\\
        00000000~00111100\\
        00000000~11110000\\
        00000000~01010101\\
        00001111~00000000\\
        00111100~00000000\\
        11110000~00000000\\
        01010101~00000000}
  &\stk{$D_{16}$\\
        0000000000001111\\
        0000000000111100\\
        0000000011110000\\
        0000001111000000\\
        0000111100000000\\
        0011110000000000\\
        1111000000000000\\
        1010101010101010} \\
    \bottomrule
  \end{tabular}
  \caption{The code-bases corresponding to the graphs\eq{e15N16}, headed with the names of the corresponding Cliffordinkras.}
  \label{t15N16}
\end{table}
\Remk
A remark about our notation is in order. The Adinkra topology types, $D_{2n}$, occurring for example in the middle row of Table~\ref{t9N14}, derive their name from the corresponding standard in the coding literature\cite{rPGMass,rMcWS,rCHVP}. In all $n\neq0\bmod4$ cases, the $D_{2n}$-type codes have $k=(n{-}1)$ generators. However, at $n=0\bmod4$, the codes in this straightforward sequence of $D_{2n}$-type become non-maximal, since an additional, weight-$2n$ generator becomes available. For $n=4$, the addition of this extra generator results in the $E_8$ code, but for $n=8,12,16,\cdots$, it is this {\em augmented\/}, maximal code that is called $D_{2n}$, rather than the non-maximal one. It would have been more reasonable to keep $D_{2n}$ for the uniform sequence:
\begin{equation}
  \Stk{$D_4$\\1111}\quad
  \Stk{$D_6$\\001111\\ 111100}\quad
  \Stk{$D_8$\\00001111\\ 00111100\\ 11110000}~\left(\&~
  \Stk{$E_8$\\00001111\\ 00111100\\ 11110000\\ 01010101}\right)\quad
  \Stk{$D_{10}$\\0000001111\\ 0000111100\\ 0011110000\\ 1111000000}\quad
  \Stk{$D_{12}$\\000000001111\\ 000000111100\\ 000011110000\\
              001111000000\\ 111100000000}\quad
  \Stk{$D_{14}$\\00000000001111\\ 00000000111100\\ 00000011110000\\
              00001111000000\\ 00111100000000\\ 11110000000000}
 \label{eD2nSeq}
\end{equation}
and so also have:
\begin{equation}
 \Stk{$D_{16}$\\
      0000000000001111\\ 0000000000111100\\ 0000000011110000\\ 0000001111000000\\
      0000111100000000\\ 0011110000000000\\ 1111000000000000}
 \qquad\text{and}\qquad
 \Stk{$E_{16}$\\
      0000000000001111\\ 0000000000111100\\ 0000000011110000\\ 0000001111000000\\
      0000111100000000\\ 0011110000000000\\ 1111000000000000\\ 0101010101010101}~,
 \label{eD*16}
\end{equation}
and so on, following the $N=8$ naming convention. In fact however, the standard nomenclature has no $E_{16}$, and instead labels:
\begin{equation}
 \Stk{No Name\\
      0000000000001111\\ 0000000000111100\\ 0000000011110000\\ 0000001111000000\\
      0000111100000000\\ 0011110000000000\\ 1111000000000000}
 \qquad\text{and}\qquad
 \Stk{$D_{16}$\\
      0000000000001111\\ 0000000000111100\\ 0000000011110000\\ 0000001111000000\\
      0000111100000000\\ 0011110000000000\\ 1111000000000000\\ 0101010101010101}
 \label{eD16}
\end{equation}
We then call the (non-maximal) code in the left-hand side of\eq{eD16} $D^*_{16}$, as it may be obtained from $D_{16}$ by omitting the largest-weight generator.

Other than that, the $E_N$-type graphs presented here share a certain `sequential resemblance', best gleaned from the corresponding graphs in\eq{e9N14} and\eq{e15N16}.
Also, we have used the multiplicative notation, as in $I^2\8D_4$, since we are primarily interested in the Adinkra topology types; for codes, this should be interpreted as the direct sum\cite{rMcWS,rCHVP}.

The fact that graphs turn out to be quite useful in listing codes is not surprising: the issue of code-(in)equivalence is closely related to graph-(in)equivalence\cite{rEPRMR}. For not too large $N$, however, innate human perception and pattern recognition then aids us in more easily discerning inequivalent graphs than the corresponding codes.

The code and Adinkra topology trees\eqs{eCode3}{eCode3n} are then extended by the maximal codes in Tables~\ref{t9N14} and~\ref{t15N16}, and all their subcodes, following the procedure discussed for the $1\leq N\leq8$ cases. There are clearly many more codes and corresponding Adinkra topologies in the $9\leq N\leq16$ and $0\leq k\leq\vk(N)$ batch than in the previous, $1\leq N\leq8$ and $0\leq k\leq\vk(N)$ cases, and we omit listing them all and assembling into the extension of the trees\eqs{eCode3}{eCode3n}.

\subsection{$N>16$}
The grahical and combinatorial complexity of the task of finding even just the maximal codes now begins to surpass the common observational pattern-recognition for {\em exhaustively and effectively\/} constructing all inequivalent graphs and so all the codes and all the Adinkra topologies. We thus seek a computer code that would do so.

\section{Conclusions}
Previous work has shown that the problem of classifying the finite-dimensional, off-shell, linear representations of $(1|N)$-supersymmetry, known as $(1|N)$-supermultiplets, may be factored\cite{r6-1}. All such representations consist of a graded vector space, spanned by an equal number of bosonic and fermionic component fields, which supersymetry transforms into each other.

The coarse classification of these representations consists of specifying the {\em topologies\/} available to the graphical rendition, Adinkras, of $(1|N)$-supermultiplets. With a topology fixed, the hanging garden theorem of Ref.\cite{r6-1} constructs, recursively, all the Adinkras, and so supermultiplets, with that topology.

In this note, we describe a relationship between the topologies available to Adinkras, graphs designed to list them, and doubly-even error-correcting $[N,k]$-codes. Here, $N$ denotes, respectively, the number of supersymmetry generators, the number of graph nodes, and the length of the binary codewords; in turn, $k$ parametrizes the size of the $(1|N)$-supermultiplets as $(2^{N-k-1};2^{N-k-1})$-dimensional, the number of certain $4n$-gon subgraphs, and the number of generators of the code. We then present the complete results for $1\leq N\leq16$, with the maximal choice, $k=\vk(N)$ in sections~\ref{s1N8} and~\ref{s9N16}. The task of listing all the non-maximal cases $0\leq k<\vk(N)$ is done for $1\leq N\leq8$ in section~\ref{s1N8}, resulting in the tree-diagrams\eqs{eCode3}{eCode3n}. The particular results pertaining to error-correcting codes have been known before\cite{rMcWS,rCHVP}, but have never been related to the iterated $\ZZ_2$-quotients of the $N$-cube, much less to the topological structure of representations of any algebra, least of all supersymmetry.

This now established connection, however, points to the inherent and formidable complexity in the task of the coarse classification, leading to a previously unsuspected abundance of representations so defined. In fact, not only is, as established here, the number of topologies available to Adinkras combinatorially growing with $N$, so is the number of Adinkras {\em per\/} topology\cite{r6-1}. Furthermore, the so constructed $(1|N)$-supermultiplets are merely the starting point in the usual recursive construction of higher-than-fundamental representations of any algebra, the simplest non-trivial infinite sequence of which was presented in Ref.\cite{r6-1}.

\bigskip\paragraph{\bf Acknowledgments:}
 The research of S.J.G.\ is supported in part by the National Science
Foundation Grant PHY-0354401.
 T.H.\ is indebted to the generous support of the Department of
Energy through the grant DE-FG02-94ER-40854.
 The Adinkras were drawn with the aid of {\em Adinkramat\/}~\copyright\,2008 by G.~Landweber.

\end{document}